\documentclass{article}
\usepackage{graphicx,wrapfig}
\usepackage{amsmath}
\usepackage{amsthm}
\usepackage{verbatim}
\usepackage[auth-sc]{authblk}
\usepackage{times}

\DeclareMathOperator{\KS}{C}
\DeclareMathOperator{\CAM}{CAM}
\DeclareMathOperator{\CD}{CD}
\DeclareMathOperator{\poly}{poly}
\DeclareMathOperator{\LDC}{LDC}
\DeclareMathOperator{\TR}{TR}
\DeclareMathOperator{\Pclass}{P}
\DeclareMathOperator{\NP}{NP}

\DeclareMathOperator{\prob}{Prob}
\DeclareMathOperator{\pr}{Pr}

\newcommand{\eps}{\varepsilon}
\def\mmid{\mathchoice{\hskip1.5pt|\hskip1.5pt}{\hskip1.5pt|\hskip1.5pt}{\hskip0.5pt|\hskip0.5pt}{\hskip0.3pt|\hskip0.3pt}}

\newtheorem{theorem}{Theorem}
\newtheorem{definition}{Defintion}
\newtheorem{lemma}{Lemma}

\begin{document}

\title{Variations on Muchnik's\\ Conditional Complexity Theorem%
        \thanks{Supported by ANR Sycomore, NAFIT ANR-08-EMER-008-01 and
        RFBR~09-01-00709-a grants.}%
}
\author[1]{Daniil Musatov}
\author[2,3]{Andrei Romashchenko}
\author[2,3]{Alexander Shen}

\affil[1]{Lomonosov Moscow State University}
\affil[2]{LIF de Marseille, CNRS \& Univ. Aix--Marseille}
\affil[3]{On leave from the  Institute for Information Transmission Problems  of RAS, Moscow.}

\maketitle

\begin{abstract}

\sloppy

Muchnik's theorem about simple conditional descriptions states
that for all strings $a$ and $b$ there exists a program $p$
transforming $a$ to $b$ that has the least possible length and
is simple conditional on~$b$. In this paper we present two new
proofs of this theorem. The first one is based on the on-line
matching algorithm for bipartite graphs. The second one, based
on extractors, can be generalized to prove a version of
Muchnik's theorem for space-bounded Kolmogorov complexity.
Another version of Muchnik's theorem is proven for a resource-bounded
variant of Kolmogorov complexity based on Arthur--Merlin protocols.

\end{abstract}

\section{Muchnik's Theorem}

In this section we recall a result about conditional Kolmogorov
complexity due to An.~Muchnik~\cite{muchnik-codes}. By $\KS(u)$
we denote Kolmogorov complexity of string $u$, i.e., the length of a
shortest program generating $u$. The conditional complexity of $u$
given $v$, the length of a shortest program that translates $v$
to $u$, is denoted by $\KS(u\mmid v)$, see \cite{li-vitanyi}.

\begin{theorem}\label{main-theorem}
Let $a$ and $b$ be two binary strings, $\KS(a)<n$ and $\KS(a\mmid b)<k$.
Then there exists a string~$p$ such that

\textbullet\ $\KS(a\mmid p,b) \le O(\log n)$;

\textbullet\ $\KS(p) \le  k+O(\log n)$;

\textbullet\ $\KS(p\mmid a) \le  O(\log n)$.
\end{theorem}
This is true for all $a,b,n,k$, and the constants hidden in $O(\log n)$ do not
depend on them. 

\textbf{Remarks}.
1. In the second inequality we can replace complexity $\KS(p)$ of a string $p$ by its length $|p|$. Indeed, we can use the shortest description of $p$ instead of $p$. 

2. We may let $k=\KS(a\mmid b)+1$ and replace $k+O(\log n)$ by $\KS(a\mmid b)+O(\log n)$ in the second inequality. We may also let $n=\KS(a)+1$.

3. Finally, having $|p|\le \KS(a\mmid b)+O(\log n)$, we can delete $O(\log n)$ last bits in $p$, and the first and third inequalities will remain true. We come to the following reformulation of Muchnik's theorem: \emph{for every two binary strings $a$ and $b$ there exist a binary string $p$ of length at most $\KS(a\mmid b)$ such that $\KS(a\mmid p,b)\le O(\log \KS(a))$ and $\KS(p\mmid a)\le O(\log\KS(a))$}.

Informally,  Muchnik's theorem says that there exists a program $p$
that transforms $b$ to $a$, has the minimal possible complexity
$\KS(a\mmid b)$ up to a logarithmic term, and, moreover, can be easily
obtained from $a$. The last requirement is crucial, otherwise
the statement becomes a trivial reformulation of the definition
of conditional Kolmogorov complexity.

This theorem is an algorithmic counterpart of Slepian--Wolf
theorem~\cite{slepian-wolf}
in multisource information theory. Assume that some person
\textbf{S} knows $b$ and wants to know $a$. We know $a$ and
want to send some message $p$ to \textbf{S} that will allow
\textbf{S} to reconstruct $a$. How long should be this message?
Do we need to know $b$ to be able to find such a message?
Muchnik's theorem provides kind of a negative answer to the last
question, though we still need a logarithmic advice.
Indeed, the absolute minimum for a
complexity of a piece of information $p$ that together with $b$
allows \textbf{S} to reconstruct $a$, is $\KS(a\mmid b)$. It is easy to
see that this minimum can be achieved with logarithmic
precision by a string $p$ that has logarithmic complexity
conditional on $a$ and $b$. But it turns out that in fact $b$ is
not needed and we can provide $p$ that is simple conditional on
$a$ and still does the job.

In many cases statements about Kolmogorov complexity have
combinatorial counterparts, and sometimes it is easy to show the
equivalence between complexity and combinatorial statements.
In the present paper we study two different combinatorial
objects closely related to Muchnik's theorem and its proof.

First, in Sect.~\ref{online}, we define the \emph{on-line
matching} problem for bipartite graphs. We formulate some
combinatorial statement about on-line matchings. This statement:
(1)~easily implies Muchnik's theorem and (2)~can be proven using
the same ideas that were used by Muchnik in
his original proof, with some adjustments.

Second, in Sect.~\ref{ext-muchnik-proof}, following
\cite{fortnow}, we use extractors and their combinatorial
properties. Based on this technique, we give a new proof of
Muchnik's theorem. With this method we prove versions of this
theorem for polynomial space Kolmogorov complexity and also for
some very special version of polynomial time Kolmogorov
complexity.

This work was presented on the  CSR2009 conference in Novosibirsk,
Russia on 18--23 August, 2009, and the conference version of the
paper was published in CSR2009 Proceedings by Springer-Verlag.
This version of the paper is slightly rearranged and extended.

\section{Muchnik's Theorem and On-line Matchings}
\label{online}

In this section we introduce a combinatorial problem that we call \emph{on-line matching}. It can be considered as an on-line version of the classical matching problem. Then we formulate some combinatorial statement about on-line matchings and explain how it implies Muchnik's theorem. Finally, we provide a proof of this combinatorial statement, starting with the off-line version of it. This finishes the proof of Muchnik's theorem.

\subsection{On-line Matchings}\label{subseq:online}

Consider a bipartite graph with the left part $L$, the right
part $R$ and a set of edges $E\subset L\times R$. Let $s$ be
some integer. We are interested in the following property of the
graph:

\begin{quote}
\emph{for any subset $L'$ of $L$ of size at most $s$ there
exists a subset $E'\subset E$ that performs a bijection between
$L'$ and some $R'\subset R$}.
\end{quote}

A necessary and sufficient condition for this property is
provided by  well-known Hall's theorem. It says that \emph{for each set
$L'\subset L$ of size $t\le s$ the set of all neighbors of
elements of $L'$ contains at least $t$ elements.}

\begin{wrapfigure}[5]{r}{30mm}
\vbox{\vspace*{-2mm}\includegraphics[scale=0.7]{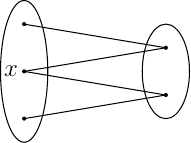}\vss}
\end{wrapfigure}
This condition is not sufficient for the following on-line
version of matching. We assume that an adversary gives us
elements of $L$ one by one, up to $s$ elements.
At each step we should provide a counterpart for each given
element $x$, i.e., to choose some neighbor $y\in R$ not used
before. This choice is final and cannot be changed later.

Providing a matching on-line, when next steps of the adversary
are not known in advance, is a more subtle problem than the
usual off-line matching. Now Hall's criterion, while still being 
necessary, is no more sufficient. For example, for the graph shown in the
picture, one can find a matching for each subset of size at most
2 of the left part, but this cannot be done on-line. Indeed, we are
blocked if the adversary starts with $x$.

Now we formulate a combinatorial statement about on-line
matching; then in Sect.~\ref{proof-muchnik} we show that this property implies Muchnik's
theorem, and in Sect.~\ref{olm-exist}
we prove this property.

\textbf{Combinatorial statement about on-line matchings}
(\textbf{OM}).
\emph{There exists a constant $c$ such that for every integers $n$
and $k$, where $k\le n$, there exists a bipartite graph $G$
whose left part $L$ has size $2^n$, right part $R$ has size $2^k
n^c$, each vertex in $L$ has at most $n^c$ neighbors in $R$, and for which
on-line matching is possible up to size $2^k$.}

Note that the size of the on-line matching is close to the size of $R$
up to a polynomial factor, and the degrees of all $L$-elements
are polynomially bounded, so we are close to Hall's bound.

\subsection{Proof of Muchnik's theorem}
        \label{proof-muchnik}

First we show how (OM) implies
Muchnik's theorem. We may assume without loss of generality that
the \emph{length} of the string $a$ (instead of its
complexity) is less than $n$. Indeed, if we replace $a$
by a shortest program that generates $a$, all complexities
involving $a$ change by only $O(\log n)$ term: knowing the
shortest program for $a$, we can get $a$ without any additional
information, and to get a shortest program for $a$ given $a$
we need only to know the value of $\KS(a)$, because we can try all programs of length $\KS(a)$
until one of them produces $a$. There may exist several
different shortest programs for~$a$; we take that one
which appears first when trying in parallel all programs of
length $\KS(a)$. As we have said, for similar reasons it does not matter whether
we speak about $\KS(p)$ or $|p|$ in the conclusion of the theorem.
We used $\KS(p)$ to make the statement more uniform; however, in the proof
we get the bound for $|p|$ directly.

We may assume that $n\ge k$, otherwise the statement of theorem~\ref{main-theorem} is trivial (let $p=a$).
Consider the graph $G$ provided by (OM) with parameters $n$ and
$k$. Its left part $L$ is interpreted as the set of all strings
of length less than $n$; therefore, $a$ is an element
of $L$. Knowing $b$, we can enumerate all strings $x$ of length less than $n$
such that $\KS(x\mmid b)<k$. There exist at most $2^k$ such strings,
and $a$ is one of them. The property (OM) implies that it
is possible to find an on-line matching for all these strings,
in the order they appear during the enumeration. Let $p$ be an
element of $R$ that corresponds to $a$ in this matching.

Let us check that $p$ satisfies all the conditions of Muchnik's
theorem. First of all, note that the graph $G$ can be chosen in
such a way that its complexity is $O(\log n)$. Indeed,
(OM) guarantees that a graph with the required properties
exists. Given $n$ and $k$, we can perform an exhaustive search
until the first graph with these properties is found. This graph
is a computable function of $n$ and $k$, so its complexity does
not exceed the complexity of the pair $(n,k)$, which is $O(\log n)$.

If $a$ is given (as well as $n$ and $k$), then $p$ can be
specified by its ordinal number in the list of $a$-neighbors.
This list contains at most $n^c$ elements, so the ordinal
number contains $O(\log n)$ bits.

To specify $p$ without knowing $a$, we give the ordinal number
of $p$ in $R$, which is $k+O(\log n)$ bits long. Here we again need $n$ and
$k$, but this is another $O(\log n)$ bits.

To reconstruct $a$ from $b$ and $p$, we  enumerate
all strings of lengths less than $n$ that have conditional
complexity (relative to $b$, which is known) less than~$k$, and
find $R$-counterparts for them using (OM) until $p$
appears. Then $a$ is the $L$-counterpart of $p$ in this
matching.

Formally speaking, for given $n$ and $k$ we should fix not only
a graph $G$ but also some on-line matching procedure, and use
the same procedure both for constructing $p$ and for
reconstructing $a$ from $b$ and $p$.\qed

\subsection{On-line Matchings Exist}
        \label{olm-exist}

It remains to prove the statement (OM). Our proof follows
the original Muchnik's argument adapted for the combinatorial
setting.

First, let us prove a weaker statement when on-line matchings
are replaced by off-line matchings. In this case the statement
can be reformulated using Hall's criterion, and we get the following
statement:

\textbf{Off-line version of (OM)}.
\emph{There exists a constant~$c$ such that for any integers $n$
and $k$, where $n>1$ and $k\le n$, there exists a bipartite graph $G$
whose left part $L$ is of size $2^n$, the right part $R$ is of
size $2^kn^c$, each vertex in $L$ has at most $n^c$ neighbors in
$R$ and for any subset $X\subset L$ of size $t\le 2^k$ the set
$N(X)$ of all neighbors of all elements of $X$ contains at least
$t$ elements.}

We prove this statement by probabilistic arguments. We choose at
random (uniformly and independently) $n^c$ neighbors for each
vertex $l\in L$. In this way we obtain a (random) graph where
all vertices in $L$ have degree at most $n^c$; the degree can be less,
as two independent choices for some vertex may coincide.

We claim that this random graph has the required property with
positive probability. If it does not, there exists a set
$X\subset L$ of some size $t\le 2^k$ and a set $Y$ of size less
than $t$ such that all neighbors of all elements of $X$ belong
to~$Y$. For fixed $X$ and $Y$ the probability of this event is
bounded by
        $
\left(\frac{1}{n^c}\right)^{t n^c}
        $
since we made $t n^c$ independent choices ($n^c$ times for each
of $t$ elements) and for each choice the probability to get into
$Y$ is at most $1/n^c$ (the set $Y$ covers at most $1/n^c$
fraction of points in $R$).

To bound the probability of violating the required property of
the graph, we multiply the bound above by the number of pairs
$X$, $Y$. The set $X$ can be chosen in at most $(2^n)^t$
different ways, since for each of $t$ elements we have at most
$2^n$ choices; actually the number is smaller since the order of
elements does not matter. For $Y$ we have at most $(2^k
n^c)^t$ choices. Further we sum up these bounds for all $t\le
2^k$. Therefore the total bound is
        $$
\sum_{t=1}^{2^k}\left(\frac{1}{n^c}\right)^{t n^c}
\left(2^{n}\right)^t
\left(2^{k}n^c\right)^t.
        $$
This is a geometric series; the sum is less than $1$ (which is
our goal) if the base is small. The base is
        $$
\left(\frac{1}{n^c}\right)^{n^c}
\left(2^{n}\right)
\left(2^{k}n^c\right)
=
\frac{2^{n+k}}{n^{c(n^c-1)}}
        $$
and $c=2$ makes it small enough for all $n>1$ and $k\le n$.
It even tends to zero as $n\to\infty$. Off-line version is proven.\qed

Now we have to prove (OM) in its original on-line version.
Fix a graph $G$ that satisfies the conditions
for the off-line version for given $n$ and $k$. Let us use the
same graph in the on-line setting with the following straightforward
``greedy'' strategy. When a new element $x\in L$ arrives, we
check if it has neighbors that are not used yet. If yes, one of
these neighbors is chosen to be a counterpart of $x$. If not,
$x$ is ``rejected''.

Before we explain what to do with the rejected elements, let us
prove that at most half of $2^k$ given elements could be
rejected. Assume that more than $2^{k-1}$ elements are rejected.
Then less than $2^{k-1}$ elements are served and therefore less
than $2^{k-1}$ elements of $R$ are used as counterparts. But all
neighbors of all rejected elements are used; this is the
only reason for rejection. So we get the contradiction with
the condition $\#N(X)\ge \#X$ if $X$ is the set of rejected
elements.

Now we need to deal with rejected elements. They are forwarded
to the ``next layer'' where the new task is to find on-line
matching for $2^{k-1}$ elements. If we can do this, then we
combine both graphs using the same $L$ and disjoint right parts
$R_1$ and $R_2$; the elements rejected at the first layer
are sent to the second one. In other terms: $(n,k)$ on-line
problem is reduced to $(n,k)$ off-line problem and $(n,k-1)$
on-line problem. The latter can then be reduced to $(n,k-1)$
off-line and $(n,k-2)$ on-line problems etc.

Finally we get $k$ levels. At each level we serve at least half
of the requests and forward the remaining ones to the next
layer. After $k$ levels of filtering only one request can be
left unserved, so one more layer is enough. Note also that we may
use copies of the same graph on all layers.

More precisely, we have proven the following statement:
\emph{Let $G$ be a bipartite graph with left side $L$ and 
right side $R$ that satisfies the
conditions of the off-line version for given $n$ and $k$. Replace
each element in $R$ by $(k+1)$ copies, all connected to the same
elements of $L$ as before. Then the new graph provides on-line
matchings up to size $2^k$.}

Note that this construction multiplies both the size of $R$ and the
degree of vertices in $L$ by $(k+1)$, which is a polynomial in $n$
factor. The statement (OM) is proven.\qed

\section{Muchnik's Theorem and
Extractors}\label{ext-muchnik-proof} In this section we present
another proof of Muchnik's theorem based on the notion of
extractors. This technique was first used in a similar situation
in~\cite{fortnow}. With this technique we prove some versions of
Muchnik's theorem for resource-bounded Kolmogorov complexity.
This result was presented in the Master Thesis of one of the
authors~\cite{musatov-diplom}.

\subsection{Extractors}

Let $G$ be a bipartite graph with $N$ vertices in the left part
and $M$ vertices in the right part. The graph may have multiple
edges. Let all vertices of the left part have the same degree
$D$. Let us fix an integer $K>0$ and a real number
$\eps>0$.

\begin{definition}
A bipartite graph $G$ is a $(K,\eps)$-extractor if for
all subsets $S$ of its left part such that $\#S\ge K$ and for
all subsets $Y$ of the right part the inequality
\begin{equation}\label{extractor}
\left|\frac{\#E(S,Y)}{D\cdot\#S}-\frac{\#Y}{M}\right|<\eps
\end{equation}
holds, where $E(S,Y)$ stands for the set of edges between $S$ and $Y$.
\end{definition}
In the sequel we always assume that $N$, $M$,
$D$, and sometimes other quantities denoted by uppercase letters are powers of $2$, and use corresponding lowercase letters ($n$, $m$, $d$, etc.) to denote their
logarithms. 
In this case the extractor may be seen as
a function that maps a pair of binary strings of length $n=\log
N$ (an index of a vertex on the left) and of length $d=\log D$ (an index
of an edge incident to this vertex) to a binary string of length
$m=\log M$ (an index of the corresponding vertex on the right).

The extractor property may be reformulated as follows: consider
a uniform distribution on a set $S$ of left-part vertices. The
probability of getting a vertex in $Y$ by taking a random
neighbor of a random vertex in $S$ is equal to
$\#E(S,Y)/(D\cdot\#S)$; this probability must be
$\eps$-close to $\#Y/M$, i.e. the probability of getting
a vertex in $Y$ by taking a random vertex in the right part.

It can be proven that for an extractor graph
a similar property holds not only for
uniform distributions on $S$, but for all distributions with
min-entropy at least $k=\log K$ (this means that no element of~$L$
appears with probability greater than $1/K$). That is, an
extractor extracts $m$ almost random bits from $n$ quasi-random
bits, with min-entropy $k$ or more, using $d$ truly random bits.
For a good extractor $m$
should be close to $k+d$ and $d$ should be small,
as well as~$\eps$. Standard
probabilistic argument shows that for all $n$, $k$ and
$\eps$ extractors with near-optimal parameters $m$ and
$d$ do exist:
\begin{theorem}\label{ext-probabilistic}
For all $K$, $N$, $M$ and $\eps$ such that $1<K\le N$, $M>0$, $\eps>0$, there exists an
$(K,\eps)$-extractor with
        $$
D=\left\lceil\max\left\{\frac
MK\cdot\frac{\ln2}{\eps^2},\
\frac1{\eps^2}\left(\ln\frac
NK+1\right)\right\}\right\rceil.
        $$
\end{theorem}
So for given $n$ and $k$ we may choose the followings values of
parameters (in logarithmic scale):
        $$
d=\log(n-k)+2\log(1/\eps)+O(1) \quad \hbox{and} \quad
m=k+d-2\log(1/\eps)-O(1).
        $$

The proof may be found in~\cite{extractor-bounds}; it is also shown
there that these parameters are optimal up to an additive term
$O(\log(1/\eps))$.

So far no explicit constructions of optimal extractors have been
invented. By saying the extractor is \emph{explicit} we mean
that there exists a family of extractors for arbitrary
values of $n$ and $k$, other parameters are computable in time
$\poly(n)$, and the extractor itself as a function of two
arguments is computable in $\poly(n)$ time. All known explicit
constructions are not optimal in at least one parameter: they
either use too many truly random bits, or not fully extract
randomness (i.e., $m\ll k+d$), or  work not for all values of $k$.
In the sequel we use the following theorem proven
in~\cite{extractor-explicit}:
\begin{theorem}
        \label{th:extractor-explicit}
For all $k$, $n$ and $\eps$ such that  $1<k\le n$ and $\eps>1/\poly(n)$, there exists an
explicit $(2^k,\eps)$-extractor with $m=k+d$ and
$d=O((\log n\log\log n)^2)$.
\end{theorem}
For the sake of brevity we use shorter and slightly weaker bound
$O(\log^3 n)$ instead of
$O((\log n\log\log n)^2)$ in the sequel.

\subsection{The Proof of Muchnik's Theorem}
\label{muchnik-with-extractor-section}

Now we show how to prove Muchnik's theorem using the extractor
technique. Consider an extractor with some $N$, $K$, $D$, $M$
and $\eps$. Let $S$ be a subset of its left part such that
$\#S \le K$. We say that a right-part element is
\emph{bad for $S$} if it has more than $2DK/M$ neighbors in $S$, that is,
twice more than the expected value if neighbors in the right part are
chosen at random and $S$ has maximal possible size $K$.
We say that a left-part element is \emph{dangerous in $S$}
if all its neighbors are bad for $S$.

\begin{lemma}\label{fortnow-lemma}
The number of dangerous elements in $S$ is less than
$2\eps K$.
\end{lemma}

\begin{proof}
We reproduce a simple proof from~\cite{fortnow}.
Without loss of generality we may  assume that $S$ contains
exactly $K$ elements; indeed, the sets of bad and dangerous elements can
only increase when $S$ increases.

For any graph, the fraction of bad right-part
vertices is at most~$1/2$, because the degree of a bad vertex is at
least twice as large as the average degree. The extractor
property reduces this bound from $1/2$ to $\eps$. Indeed,
let $\delta$ be the fraction of bad elements in the right part.
Then the fraction of edges going to bad elements (among all
edges starting at $S$) is at least $2\delta$. Due to the
extractor property, the difference between these fractions
should be less than $\eps$. The inequality
$\delta<\eps$ follows.

Now we count dangerous elements in $S$. If their fraction in $S$
is $2\eps$ or more, then the fraction of edges going to the
bad elements (among all edges leaving $S$) is at least
$2\eps$. But the fraction of bad vertices is less
than $\eps$, and the difference between two fractions
should be less than $\eps$ due to the extractor
property.\qed
\end{proof}

Now we present a new proof of Muchnik's theorem. As we have seen
before, we may assume without loss of generality that the length
of $a$ is less than $n$. Moreover, as we have said, we may assume that conditional
complexity $\KS(a\mmid b)$ equals $k-1$ (otherwise we decrease $k$)
and that $k<n$ (otherwise the theorem is obvious, take $p=a$).

Consider an extractor with given $n$, $k$; let $d=O(\log n)$,
$m=k$ and $\eps=1/n^3$; such an extractor exists due to
Theorem~\ref{ext-probabilistic}. The choice of $\eps$
will become clear later. We choose an extractor whose
complexity is at most
$2\log n+O(1)$. It is possible, because only $n$ and $k$ are needed to describe
such an extractor: other parameters are functions of $n$ and
$k$, and we can search through all bipartite graphs with given
parameters in some natural order until the first extractor with
required parameters is found. This search requires a very long
time, so this extractor is not explicit.

Now assume that an extractor is fixed. We treat the left part of
the extractor as the set of all binary strings of length less
than $n$ (including $a$), and the right part as the set of all
binary strings of length $m=k$ (we will choose $p$ among them).
Consider the set $S_b$ of all strings in the left part such that
their complexity conditional on $b$ is less than $k$; note that $a$
belongs to this set.

We want to apply Lemma~\ref{fortnow-lemma} to the set $S_b$ and
prove that $a$ is not dangerous in $S_b$
by showing that otherwise $\KS(a\mmid b)$ would
be too small. So $a$ has a neighbor $p$ that is not
bad for $S_b$, and this $p$ has the required properties.

According to this plan, let us consider two cases.

\textbf{Case 1}. If $a$ is not dangerous in $S_b$,  then $a$ has a
neighbor $p$ that is not bad for $S_b$. Let us show that $p$
satisfies the claim of the theorem.

Complexity of $p$ is at most $k+O(1)$ because its length is $k$.

Conditional complexity $\KS(p\mmid a)$ is logarithmic because $p$ is
a neighbor of $a$ in the extractor and to specify $p$ we need a
description of the extractor ($2\log n+O(1)$ bits) and the
ordinal number of $p$ among the neighbors of $a$ ($d=\log
D=O(\log n)$ bits).

As $p$ is not bad for $S_b$, it has less than $2D$ neighbors in
$S_b$. If $b$ is known, the set $S_b$ can be
enumerated; knowing $p$, we select
neighbors of $p$ in this enumeration.
Thus, to describe $a$ given $p$ and $b$, we need only a
description of the extractor and the ordinal number of $a$ in
the enumeration of the neighbors of $p$ in $S_b$, i.e., $O(\log
n)$ bits in total.

\textbf{Case 2}. Assume that $a$ is dangerous in $S_b$. Since
the set $S_b$ can be enumerated given $b$, the sets of all bad vertices (for
$S_b$) and all dangerous elements in $S_b$ can also be enumerated.
Therefore, $a$ can be specified by the string $b$,
the extractor and the ordinal number of $a$ in the enumeration of all dangerous
elements in $S_b$.
This ordinal number consists of $k-3\log n+O(1)$ bits due to
the choice of $\eps$ (Lemma~\ref{fortnow-lemma}).
So, the full description of $a$
given $b$ consists of $k-\log n+O(\log\log n)$ bits; $O(\log\log n)$ additional bits are needed for separating
$n$, $k$ and the ordinal number.
This contradicts the
assumption that $\KS(a\mmid b) = k-1$. Thus, the second case is
impossible and Muchnik's theorem is proven.\qed

\subsection{Several Conditions and Prefix Extractors}

In~\cite{muchnik-codes} An.~Muchnik proved also the following
generalization of Theorem~\ref{main-theorem}:
\begin{theorem}\label{two-conditions-theorem}
Let $a$, $b$ and $c$ be binary strings, and let $n$, $k$ and $l$
be numbers such that $\KS(a)<n$, $\KS(a\mmid b)<k$ and $\KS(a\mmid c)<l$.
Then there exist binary strings $p$ and $q$ of length $k$ and
$l$ respectively such that one of them is a prefix of the other one and all the
conditional complexities $\KS(a\mmid p,b)$, $\KS(a\mmid q,c)$, $\KS(p\mmid a)$,
$\KS(q\mmid a)$ are of order $O(\log n)$.
\end{theorem}
This theorem is quite non-trivial: indeed, it says that
information about $a$ that is missing in $b$ and $c$ can be
represented by two strings such that one is a prefix of the
other, even if $b$ and $c$ are completely unrelated. It
implies also that for every three strings $a,b,c$ of length less
than $n$, the minimal length of a program that transforms $b$
to $a$ and at the same time transforms $c$ to $a$ is at most
$\max\{\KS(a\mmid b),\KS(a\mmid c)\}+O(\log n)$.

In fact a similar statement can be proven not only for two but
for many (even for $\poly(n)$) conditions. For the sake of
brevity we consider only the statement with two conditions.

This theorem also can be proven using extractors. Any extractor
can be viewed as a function
$E\colon\{0,1\}^n\times\{0,1\}^d\to\{0,1\}^m$.
\begin{definition}
We say that a
$(2^k,\eps)$-extractor $E\colon\{0,1\}^n\times\{0,1\}^d\to\{0,1\}^m$,
where $m\ge k$,
is a \emph{prefix extractor} if
for every $i\le k$ its prefix of length $m-i$, i.e., a
function $E_i\colon\{0,1\}^n\times\{0,1\}^d\to\{0,1\}^{m-i}$
obtained by truncating $i$ last bits, is a
$(2^{k-i},\eps)$-extractor.
 \end{definition}
By using probabilistic method
the following theorem can be proven:
\begin{theorem}\label{prefix-extractor-theorem}
For all $k$, $n$, and $\eps$ such that $1<k\le n$ and $\eps>0$, 
there exists a prefix
$(2^k,\eps)$-extractor with parameters $d=\log
n+2\log(1/\eps)+O(1)$ and
$m=k+d-2\log(1/\eps)-O(1)$.
\end{theorem}
\textbf{Proof:}
This proof is quite similar to the standard proof of
Theorem~\ref{th:extractor-explicit}. In that proof the probabilistic argument is used to
show that a random graph has the required property with positive probability.
In fact it is shown that this probability is not only positive but close to $1$.
Then we note that the restriction of a random graph is
also a random graph, and the intersection of several events having
probability close to $1$ has a positive probability. Let us explain these arguments
in more detail.

We want to show that a random bipartite graph with
given parameters is a prefix extractor with a
positive probability.
First of all we note that it is enough to show that
inequality~(\ref{extractor}) holds for $S$ of size exactly $K$.
Then this condition is true also for every bigger set $S$, since the uniform distribution on $S$ is an average of the distributions on its subsets of size $K$. Second,
it is enough to check the bound (\ref{extractor}) only in
one direction:
       $$
  \frac{\#E(S,Y)}{D\cdot \#S} < \frac{\#Y}{M}   + \eps
       $$
for all sets $S$ of cardinality $K$  and for all $Y$.
Indeed, the inequality
 $$
  \frac{\#E(S,Y)}{D\cdot \#S} > \frac{\#Y}{M}   - \eps
 $$
follows from the previous one applied to the complement of $Y$:
if there are too few edges from $S$ to $Y$ then there are too
many edges from $S$ to the complement of $Y$.

Now we specify the distribution on graphs. For every
string of length $n$ (a vertex of the left part) we choose at
random (uniformly and independently) $D=2^d$ strings of length
$m$ (its neighbors in the right part). Now we bound the probability of
the event \emph{a random graph is not a prefix extractor}.

If the extractor property is violated for some prefix of length
$m-i$ then there exists a set $S$ of $K/2^{i}$ elements from the
left part and a set $Y\subset\{0,1\}^{m-i}$ of size $\alpha
2^{m-i}$ (for some $\alpha>0$) such that the number of edges
between $S$ and $Y$ is greater than $(\alpha+\eps)KD/2^{i}$. From
the Chernoff-Hoeffding bound it follows that probability of this event is
not greater than $\exp(-2\eps^2KD/2^{i})$. Hence,
probability of the event \emph{a random graph is not a prefix
extractor} can be limited by the sum of such bounds for all $i$,
$S$, and $Y$:
        $$
\sum\limits_{i=0}^k \left(\begin{array}{c}N\\K/2^{i}\end{array}\right)
  \cdot2^{M/2^{i}}\exp(-2\eps^2KD/2^{i}).
        $$
Since $\binom{u}{v} \le u^v/v! \le (ue/v)^v$, this sum does not exceed
        \begin{multline*}
\sum\limits_{i=0}^k \left(\frac{eN}{K/2^i}\right)^{K/2^i}
2^{M/2^i}\exp(-2\eps^2KD/2^i)=\\
=\sum\limits_{i=0}^k\left(e^{(K/2^i)(1+\ln(2^iN/K))}\cdot
e^{-\eps^2KD/2^i}\right)\cdot\left(e^{M\ln2/2^i}\cdot
e^{-\eps^2KD/2^i}\right).
        \end{multline*}
The condition of the theorem implies that $D\ge
\frac{M}{K}\cdot\frac{\ln2}{\eps^2}$, assuming that $O(1)$ constant is large enough. Hence, the second factor
in each term of the sum is not greater than $1$. On the other hand,
the first factor equals
        $$
e^{(K/2^i)(1+\ln(2^iN/K)-\eps^2D)}\le e^{(K/2^i)(1+\ln N-\eps^2D)},
        $$
which is less than $(1/2)^{(K/2^i)}$, since $D\eps^2\ge
1+\ln2+\ln N$. The sum of these terms is strictly less than
$1$. Thus, probability of the event \emph{a random graph is a
prefix extractor} must be positive. \qed

\smallskip

However, using prefix extractors is not enough; we need to
modify the argument, since now we need to find two related
neighbors in two graphs. So we modify the notion of a dangerous
vertex and use the following analog of
Lemma~\ref{fortnow-lemma}:
\begin{lemma}\label{generalized-fortnow-lemma}
Let us call a left-part element weakly dangerous in $S$ if at
least half of its neighbors are bad for $S$. Then the number of
weakly dangerous elements in $S$ is at most $4\eps K$.
\end{lemma}
\textbf{Proof:} is similar to the proof of Lemma~\ref{fortnow-lemma}.
Since only half of all neighbors are bad, we need twice more
elements.\qed

\smallskip

Now we give a new proof of
Theorem~\ref{two-conditions-theorem} based on prefix
extractors. Fix a prefix extractor $E$ with parameters $n$, $k$,
$d=O(\log n)$, $m=k$ and $\eps=1/n^3$. Again, we may
assume that complexity of this extractor is $2\log n+O(1)$. We
also may assume that $\KS(a\mmid b)=k-1$, $\KS(a\mmid c)=l-1$ and (without
loss of generality) $k\ge l$.

Let $S_b$ and $S_c$ be the sets of strings of conditional
complexity less than $k$ and $l$ conditional on $b$ and $c$
respectively. Call an element \emph{weakly dangerous in $S_b$}
if it is weakly dangerous (in $S_b$) for the original extractor
and \emph{weakly dangerous in $S_c$} if it is weakly dangerous
(in $S_c$) for the $l$-bit prefix of $E$. Since this prefix $E_{k-l}$
is also an extractor, the statement of
Lemma~\ref{generalized-fortnow-lemma} holds for $S_c$. The
string $a$ belongs to the intersection of $S_b$ and $S_c$ and is
not weakly dangerous in both. Hence, a random neighbor of $a$ and its
prefix are not bad for $S_b$ [resp. $S_c$] with probability
greater than $1/2$. So we can find a $k$-bit string $p$ such
that $p$ and its $l$-bit prefix $q$ are not bad for $S_b$ and
$S_c$ respectively.

They satisfy the requirements. Indeed, the conditional
complexities $\KS(p\mmid a)$ and $\KS(q\mmid a)$ are logarithmic because
$p$ and $q$ can be specified by their ordinal numbers among the
neighbors of $a$ in the extractor. The string $a$ may be
obtained from $p$ and $b$ with logarithmic advice because $p$ is not bad for $S_b$ in
$E$; similarly, $a$ can be obtained from $q$ and $c$ with logarithmic advice because $q$
is not bad for $S_c$ in $E_{k-l}$. This completes the proof of
Muchnik's theorem for two conditions.\qed

\subsection{Muchnik's Theorem about Space-Bounded Complexity}

The arguments from Sect.~\ref{muchnik-with-extractor-section}
together with constructions of explicit extractors imply some
versions of Muchnik's theorem for resource-bounded Kolmogorov
complexity. In this section we present such a theorem for the
space-bounded complexity.

First of all, the definitions. Let $\varphi$ be a multi-tape
Turing machine that transforms pairs of binary strings to binary
strings. Conditional complexity $\KS^{t,s}_{\varphi}(a\mmid b)$ is the
length of the shortest $x$ such that $\varphi(x,b)$ produces $a$
in at most $t$ steps using space at most~$s$.
It is known (see~\cite{li-vitanyi}) that there exists an
\emph{optimal description method} $\psi$ in the following sense:
for every $\varphi$ there exists a constant $c$ such that
        \begin{equation*}
\KS_{\psi}^{ct\log t,cs}(a\mmid b)\le \KS_{\varphi}^{t,s}(a\mmid b)+c.
        \end{equation*}
We fix such a method $\psi$, and in the sequel use notation
$\KS^{t,s}$ instead of $\KS_{\psi}^{t,s}$.

Now we present our variant of Muchnik's theorem for
space-bounded Kolmogorov complexity:
\begin{theorem}\label{muchnik-space}
Let $a$ and $b$ be binary strings and $n$, $k$ and $s$ be numbers
such that $\KS^{\infty,s}(a)<n$ and $\KS^{\infty,s}(a\mmid b)<k$. Then
there exists a binary string $p$ such that

\begin{itemize}
\item[\textbullet] $\KS^{\infty,O(s)+\poly(n)}(a\mmid p,b) = O(\log^3 n);$

\item[\textbullet] $\KS^{\infty,O(s)}(p)\le k+O(\log n);$

\item[\textbullet] $\KS^{\infty,\poly(n)}(p\mmid a)= O(\log^3 n),$

\end{itemize}
where all constants in $O$- and $\poly$-notation depend only on
the choice of the optimal description method.
\end{theorem}

\begin{proof}

The proof of this theorem starts as an effectivization of the
argument of Sect.~\ref{muchnik-with-extractor-section}. To find $p$
effectively, we use an explicit extractor with parameters
$n$, $k$, $d=O(\log^3 n)$,
$m=k$ and $\eps=1/n^3$. We increase  $d$ and
respectively the conditional complexity of $p$ when $a$ is given
from $O(\log n)$ to $O(\log^3 n)$, because
currently known explicit extractors use
more random bits than the ideal extractors from Theorem~\ref{ext-probabilistic}.\footnote{%
	\emph{Note added in proof}. Using Nisan--Wigderson construction of 
	pseudorandom bit generator
	one may improve this result and replace $\log^3$ by $\log$, 
	as in the original Muchnik's theorem. This argument will be published
	elsewhere.}

First we prove a weaker version of the theorem
assuming that the value of $s$ is added
as a condition (in three complexities that are bounded by the theorem).
Later we explain how to get rid of this restriction.

Recall that a right-part element is bad if it has more than $DK/M$ neighbors on the left and a left-part element is dangerous if all its neighbors are bad. Let us show that if $a$ is not dangerous and $p$ is a neighbor of $a$ that is not bad,
then we can recover $a$ from $b$ and $p$ using
$O(\log^3 n)$ extra bits of information and $O(s)+\poly(n)$
space. For any string $a'$ we can test in
$O(s)+\poly(n)$ space whether $\KS^{\infty,s} (a'\mmid b)<k$:
We test sequentially all programs of length less than $k$
and check if they produce $a'$ on space $s$ given $b$.
Simulating every such a program, we limit its workspace to $s$,
and prevent infinite loops by counting the number of steps.
If a program makes more than $c^s$ steps in space $s$ then it
loops; here  $c$ is some constant that depends only on the choice of the
universal Turing machine. This counter uses only $O(s)$ space.
Therefore, given $b$ and $p$ we can enumerate all the strings $a'$
that are neighbors of $p$ and $\KS^{\infty,s} (a'\mmid b)<k$, and
wait until a string with a given ordinal number appears.

The difficulty arises when we try to prove that $a$ is not
dangerous. Let us try to repeat our arguments taking into account the space
restrictions. First we note that one can enumerate (or recognize: for space
complexity it is the same) all bad elements in the right part using space
$O(s)+\poly(n)$. As before, we assume here that $s$ is given in addition
to $n$, $k$, and $b$. Indeed,
bad elements (as defined above) have many neighbors among strings $a'$
such that $\KS^{\infty,s} (a'\mmid b)<k$, and those strings can be enumerated.

Therefore, we can also enumerate
all dangerous elements in the left part using space $O(s)+\poly(n)$.
We know also that the number of dangerous elements is small,
but this does not give us a contradiction (as it did before) since the space used
by this enumeration
increases from $s$ to $O(s)+\poly(n)$, and even a small increase
destroys the argument. So we cannot claim that $a$ is not dangerous
and need to deal somehow with dangerous elements.

To overcome this difficulty, we use the same argument as in
Sect.~\ref{olm-exist}.  We treat the dangerous elements at the next layer,
with reduced $k$ and other extractor graph. We need
$O(k)$ layers (in fact even $O(k/\log n)$ layers) since by
Lemma~\ref{fortnow-lemma}
at every next layer the number of dangerous elements that still need to
be served is reduced at least by the factor $2\eps$.
Note also that the
space overhead needed to keep the accounting information is
$\poly(n)$ and we never need to run in parallel several
computations that require space $s$; this space is needed only at the
bottom level of the recursion, in all other cases $\poly(n)$ is enough.

So we get the theorem in its weak form (with condition $s$). For the
full statement some changes are needed. Let us sequentially use space bounds
$s'=1,2,\ldots$: to enumerate all strings $a'$ such that
$\KS^{\infty,s} (a'\mmid b)<k$, we sequentially enumerate all strings
that can be obtained from $b$ and a $k$-bit encoding using space
$s'=1,2$, etc. The corresponding set increases as $s'$
increases, and at some point we enumerate all strings $a'$ such
that $\KS^{\infty,s} (a'\mmid b)<k$, though this moment is not known
to us. Note that we can avoid multiple copies of the same
string for different values of $s'$:
performing the enumeration for $s'$, we check for every
string whether it has appeared earlier, using $s'-1$ instead of
$s'$. This requires a lot of time, but only $O(s)$ space.
Knowing the ordinal number of $a$ in the entire enumeration, we
stop as soon as it is achieved; hence, the enumeration process
requires only space $O(s)+\poly(n)$, though $s$ is not specified
explicitly.

Similarly, the set of dangerous strings $a$ (that go to the second or
higher layer) increases as $s'$ increases, and can be enumerated
sequentially for $s'=1,2,3\ldots$ without repetitions in
$O(s')+\poly(n)$ space. Therefore, at every layer we can use the
same argument, enumerating all the elements that reach this
layer and at the same time are neighbors of $p$, until we produce
as many of them as required.\qed
\end{proof}

\textbf{Remarks}.
1.~The process of enumerating $a'$
such that $\KS^{\infty,s'}(a'\mmid b)<k$
sequentially for $s'=1,2,3,\ldots$ can be considered as the enumeration of
all $a'$ such that $\KS(a'\mmid b)<k$. So we just get the proof
for the unrestricted version of Muchnik's theorem with an additional
remark: if an explicit extractor is used, then the short programs provided
by this theorem require  only slightly more space than the programs given
in the condition.

2.~When we use several layers (instead of a contradiction with the
assumption that the complexity $\KS(a\mmid b)$ is exactly $k-1$) we in fact do 
not need  $\eps$ to be as small as $1/n^3$; it is enough 
to use a small constant value of $\eps$.

\subsection{Muchnik's Theorem for $\CAM$-complexity}

The arguments from the previous sections cannot be applied for
Kolmogorov complexity with polynomial \emph{time} bound. Roughly
speaking, the obstacle is the fact that we cannot implement an
exhaustive search over the list of `bad' strings in polynomial
time unless $\Pclass=\NP$. The best result that we can prove for
poly-time bounded complexity involves a version of Kolmogorov
complexity introduced in~\cite{BLvM}:
\begin{definition}\label{am-complexity}
Let $U_n$ be a non-deterministic universal Turing machine.
Arthur-Merlin complexity $\CAM^{t}(x\mmid y)$ is the length of a
shortest string $p$ such that
       \begin{enumerate}
 \item ${\rm Prob}_r[U_n(y,p,r) \mbox{ can print } x 
  \mbox{ and cannot print any other string }
 ]>2/3$

 \item $U_n(y,p,r)$ stops in  time at most $t$
 \textup(for all branches of non-deterministic computation\textup).
       \end{enumerate}
\end{definition}
As always, $\CAM^{t}(x):=\CAM^{t}(x\mmid \lambda)$.

This definition is typically used for $t=\poly(|x|)$.
Intuitively, a $\CAM$-des\-crip\-tion $p$ of a string $x$ given
another string $y$ is an interactive Arthur--Merlin protocol:
Arthur himself can do probabilistic polynomial computations, and
can ask questions to all-powerful but not trustworthy Merlin;
Merlin can do any computations and provide to Arthur any
requested certificate. So, Arthur should ask such questions that
the certificates returned by Merlin could be effectively used to
generate $x$.
With this version of resource-bounded Kolmogorov complexity we
have a variant of Muchnik's theorem:
\begin{theorem}\label{amd-theorem}
For every polynomial $t_1$, there exists a polynomial $t_2$ such
that the following condition holds. Let  $a,b$ be strings such that 
$\KS^{t_1(n),\infty}(a\mmid b)< k$, where $n=|a|$. 
Then there exists a string $p$ of length $k+O(\log^3 n)$ such that
 \begin{itemize}
 \item $\KS^{t_2(n),\infty}(p\mmid a)=O(\log^3 n)$ and
 \item $\CAM^{t_2(n)}(a\mmid b,p)=O(\log^3 n)$.
 \end{itemize}
\end{theorem}
\textbf{Proof:}
In the proof of this theorem we cannot use an arbitrary
effective extractor. We employ very essentially  properties
of one particular extractor constructed by
L.~Trevisan~\cite{trevisan}. Our arguments mostly repeat the
proof of Theorem~3 from~\cite{BLvM}.

First of all we remind the definition of the Trevisan extractor,
which is based on the technique from the seminal paper by Nisan and
Wigderson~\cite{NW97}. The first crucial ingredient of the Trevisan
function is a weak design. A system of sets
        $$S_1,\ldots,S_m\subset\{1,\ldots,d\}$$
is called a weak design with parameters $(l,d)$ if each $S_i$
consists of $l$ elements and for every $i>1$ the sum
$\sum\limits_{j=1}^{i-1}2^{\#(S_i\cap S_j)}$ is bounded by $(m-1)$.
Weak designs exist; moreover, they can be constructed
effectively. More precisely, there exists an algorithm
that for any given $l,m$ generates a week design with
$d=O(l^2\log m)$ in time polynomial in $l$ and $m$, 
see~\cite{NW97}.

Let us fix a weak design as above. For $x\in\{0,1\}^d$ we use
the following notation: $x|_{S_i}$ denotes the $l$-bit string
that is obtained by projecting $x$ onto coordinates specified by
$S_i$.

The second important ingredient of Trevisan's construction is an
error correcting code. For every positive integer $n$ and
$\delta>0$, there exists a list decodable code
         $$
  \LDC_{n,\delta}:\ \{0,1\}^n\to \{0,1\}^{\bar n}
         $$
where $\bar n = \poly(n/\delta)$, such that
 \begin{enumerate}
 \item $\LDC_{n,\delta}(x)$ can be computed in polynomial time;
 \item given any $y\in\{0,1\}^{\bar n}$, the list of all
 $x\in\{0,1\}^n$ such that $\LDC_{n,\delta}(x)$ and $y$
 agree in at least $(1/2+\delta)$ fraction of bits, can be
 generated in time $\poly(n/\delta)$. In particular, this property means
 that the number of words $x$ in  this list is not greater than $\poly(n/\delta)$;
 \end{enumerate}
(see, e.g.,~\cite{Su97}). In the sequel we will
assume that $\bar n$ is a power of $2$.

Let us fix an encoding as above and denote $l(n)=\log \bar n$.
For $u\in\{0,1\}^n$ the value $\LDC_{n,\delta}(u)$ is a string of
length $2^l$. So, we can view $\LDC_{n,\delta}(u)$ as a Boolean
function
      $$\hat u:\{0,1\}^l\to \{0,1\}$$
Having fixed a weak design $S_1,\ldots,S_m$ and an
encoding $\LDC_{n,\delta}$, we define the Trevisan function
$\TR_{\delta}: \{0,1\}^n\times\{0,1\}^d \to \{0,1\}^m$ as
        $$
 \TR_{\delta}(u,y) = \hat u (y|_{S_1})\ldots \hat u (y|_{S_m}).
        $$
We do not need to show that $\TR$ is an extractor (for suitable
values of $n, d, m$); in our proof we refer directly to
the definition of this function. We will use the Trevisan function
for $\delta=\frac{1}{8m}$ and
$m=k+d+1$. 
More precisely, the parameters are chosen as follows. Numbers $k$
and $n$ are taken from the statement of the theorem; 
$l(n)=\log \bar n$ is obtained from the construction of $\LDC_{n,\delta}$;
further, we can choose appropriate $m$ and $d=O(l^2\log m)=O(\log^3 n)$ 
so that (i) there exists a weak design with parameters $m,l,d$, and
(ii) it holds $m=k+d+1$.

Denote by $L_b$ the set of all strings whose time-bounded
complexity conditional on $b$ is less than $k$:
       $$
  L_b = \{ u\in\{0,1\}^n\ |\ C^{t_1(n),\infty}(u\mmid b)< k \}.
       $$
Then $a\in L_b$ and $\#L_b< 2^{k}$.
We have chosen such an $m$ that the $\TR$-image of $L_b\times\{0,1\}^d$ 
covers at most $50\%$ of the set $\{0,1\}^m$.
Denote by $B$ the predicate \emph{being in the $\TR$-image of
$L_b\times\{0,1\}^d$}. For every $u\in L_b$
       $$
 \prob_{r_1\ldots r_d}[B(\TR_{\delta}(u,r_1\ldots r_d))=1]-
  \prob_{r_1\ldots r_m}[B(r_1\ldots r_m)=1]\ge 1/2
       $$
since the first probability is equal to $1$ and the second one
is not greater than $1/2$. In other notation, we have
 \begin{eqnarray*}
 \prob_{y\in\{0,1\}^d}[B(\hat u(y|_{S_1}) \hat u(y|_{S_2})
       \ldots \hat u(y|_{S_m}) )=1] - \ \ \ {}\\
    -\prob_{r_1\ldots r_m}[B(r_1\ldots r_m)=1]\ge 1/2.
 \end{eqnarray*}
We apply the standard `hybridization' trick: we note that for
some $i$,
 \begin{equation}       
 \begin{split}
 \prob_{y,r_{i+1},\ldots,r_m}[B(\hat u(y|_{S_1}) \hat
     u(y|_{S_2}) \ldots \hat u(y|_{S_i})r_{i+1}\ldots r_m)=1]
     -\hskip55pt\\
    - 
   \prob_{y,r_i,r_{i+1},\ldots,r_m}
     [B(\hat u(y|_{S_1}) \hat u(y|_{S_2}) \ldots \hat u(y|_{S_{i-1}}) r_i \ldots r_m    )=1]
    \ge 1/(2m).
 \end{split}   \label{hybr-ineq}
 \end{equation}
Further, we can somehow fix the bits of $y$ outside $S_i$
so that~(\ref{hybr-ineq}) remains true. Denote $y|_{S_{i}}$
by $x$. Now each function $\hat u(y|_{S_j})$ depends on $\#(S_j\cap
S_i)$ bits from $x$ (the other bits of $y$ are fixed).  We denote this function
by $f_j$. The truth  table  of 
$$f_j\ :\ \{0,1\}^{\#(S_j\cap S_i)} \to \{0,1\}$$
consists of  $2^{\#(S_j\cap S_i)}$ bits. 
With a slight abuse of notations we will write $f_j(x)$ (though 
$f_j$  depends on only $\#(S_j\cap S_i)$ bits of string $x$).
Note that the definition of $f_j$ involves
implicitly the string $u$ and those bits in $y$ 
that we fixed outside positions $S_i$.

To specify the truth tables of all functions $f_1,\ldots, f_{i-1}$ 
we need
         \begin{equation}
  \sum\limits_{j=1}^{i-1}2^{\#(S_j\cap S_i)}<m\label{length-p}
         \end{equation}
bits (the last inequality follows from the definition of weak
designs). 

The explained construction of functions $f_1,\ldots, f_{i-1}$ works for every $u\in \{0,1\}^n$.
We denote by $p$ the concatenation of  the truth tables of  $f_1, \ldots, f_{i-1}$
for $u=a$, where $a$ is the string from the statement of the theorem.
By~(\ref{length-p}) the length of $p$ is less than $m$.

To specify this $p$ given $a$, we need to know only $m$, $i$ and the bits of $y$
fixed outside  $S_i$. Hence, $\KS^{\poly(n),\infty}(p\mmid a)=O(\log^3 n)$.

In the rest of the proof we show that there exists an
Arthur--Merlin protocol that reconstructs $a$ given $b$, $p$ and
some small additional information. Since $u=a$, it is enough
to reconstruct  string $\hat u$ (then we apply the decoding
procedure and find $a=\LDC^{-1}_{n,\delta}(\hat u)$).

Let us investigate  inequality (\ref{hybr-ineq}). To make the
notations more concise, we denote
         $$
 g_{r_i}(x,r_{i+1} \ldots r_m)= \left\{
   \begin{array}{cl}
    r_i& \mbox{ if }B(f_1(x) \ldots f_{i-1}(x)r_i\ldots r_m)=1\\
    1-r_i&\mbox{ otherwise}.
   \end{array}
   \right.
         $$
By a standard argument from the computational
XOR Lemma~\cite{Gol95} we get
        \begin{equation}\label{approx-ineq}
\prob_{x\in\{0,1\}^l,r_i \ldots r_m\in\{0,1\}^{m-i+1}}[\hat u(x) =
g_{r_i}(x,r_{i+1} \ldots r_m)]\ge 1/2+1/(2m).
        \end{equation}
Now we fix a value of $r_i$ (set it to $0$ or $1$) so that
inequality~(\ref{approx-ineq}) remains true. This bit must be included into
the description of $a$ given $b$ and $p$. Without any loss of generality we assume
that $r_i=1$, and in the sequel we omit $r_i$ in our notations.
In other words, instead of $g_{r_i}(x,r_{i+1} \ldots r_m)$ we write
         $$
 g(x,r_{i+1} \ldots r_m)= \left\{
   \begin{array}{cl}
    1& \mbox{ if }B(f_1(x) \ldots f_{i-1}(x)1 r_{i+1}\ldots r_m)=1\\
    0&\mbox{ otherwise}.
   \end{array}
   \right.
         $$

If the word $p$ defined above and a ``typical'' sequence
$r_{i+1}\ldots r_{m}$ are given,
Arthur can approximate $\hat u$ and then reconstruct $a$ (using
decoding algorithm for $\LDC_{n,\delta}$). So, Arthur chooses at random 
several copies of $r_{i+1}\ldots r_{m}$ and tries to approximate
$\hat u$ with each copy.  Further we explain how it works.

First, we need some notation.
We say a string $v'\in\{0,1\}^{\bar n}$ is an $\alpha$-ap\-pro\-xi\-ma\-tion 
to a string $v\in\{0,1\}^{\bar n}$ if  these strings coincides  in at least
$\alpha \bar n$ bits. In particular, we will be  interested in
$\alpha$-approximation to $\hat u$.

If we fix in $g(x,r)$ the second argument $r$, we get some Boolean 
function $g^{(r)}(x)$ that depends on $x\in\{0,1\}^l$.
For every fixed $r$ we identify the corresponding function $g^{(r)}(x)$
with its truth table, i.e.,
with the string $z^{(r)}$ of length $\bar n = 2^{l}$ where every $x$-th bit
equals $1$ iff $g(x,r)=1$. So, the number of $1$'s in $z^{(r)}$ is
equal to the number of strings $x$ such that 
$B(f_1(x)\cdots  f_{i-1}(x)1r)=1$.

We say that a string $v\in\{0,1\}^{\bar n}$ is a \emph{candidate} if 
$v$ is a codeword of $\LDC_{n,\delta}$, and
for at least
$1/32m$ of all $r\in \{0,1\}^{m-i}$ the corresponding string $z^{(r)}$ is an
$(1/2+1/8m)$-approximation to $v$.
From the decoding
property of the code $\LDC_{n,\delta}$, each $z\in\{0,1\}^{\bar n}$ can be an
$(1/2+1/8m)$-approximation for at most $q=\poly(m)$ different codewords
$\LDC_{n,\delta}(u)$.
Hence, there exist at most
$32mq$ candidates (of course, $\hat u$ is a candidate). By
Sipser's $\CD$-coding theorem~\cite{Si83}
there exists a poly-time program
$p'$ of length $2\log (32mq) = O(\log n)$ that accepts $\hat u$
and rejects all other candidates (no warranty about
non-candidates: $p'$  may accept or reject any of them).

\textbf{First part of the Arthur--Merlin protocol}:
Denote
         $$
  \bar g = \sum\limits_{x,r}g(x,r)/2^{m-i}.
         $$
This is the average number of strings $x\in\{0,1\}^{l}$ such that $g(x,r)=1$
for a random $r\in\{0,1\}^{m-i}$.

At first Arthur chooses $s$ random strings $r(1),\ldots,r(s)$ of
length $(m-i)$ (a polynomial $s=s(n)$ is specified below).
He asks Merlin to generate $s\cdot (\bar g-\gamma)$
($\gamma=\gamma(n)$ is also specified below) 
certificates for the facts that different tuples 
$\langle x,r(j)\rangle$ satisfy
$g(x,r(j))=1$, and verifies these certificates. 

Indeed, if $B(w)=1$ for some string $w$, Merlin can provide a
certificate for this fact: he communicates to Arthur 
(i)~some $u,y$ such that $\TR_{\delta}(u,y)=w$, and (ii)~provides a
poly-time program $\pi$ of length less than $k$ such that $\pi(b)$
stops in $t_1$ steps and returns $u$; that is, Merlin proves to
Arthur that $u\in L_b$.

If
at least one certificate is false, Arthur stops without any
answer. If the certificates are OK, Arthur calculates $\tilde z_{1},
\ldots, \tilde z_{s}$, where $x$-th bit of $\tilde z_{j}$ is $1$ iff Merlin
provided a certificate of the fact that $g(x,r(j))=1$.

We need the following probabilistic lemma:
 \begin{lemma}\label{lemma1}
For some rational $\gamma=\bar n/\poly(m)$ and integer $s=\poly(n)$,
for randomly chosen strings $r(1),\ldots,r(s)$ of length $m-i$,
with probability more than $2/3$
Merlin can provide some certificates for at least $s\cdot (\bar g-\gamma)$
strings $r(j)$ \textup(each certificate must prove for one of  $r(j)$
that $g(x,r(j))=1$\textup), and,
whatever  certificates are chosen by Merlin, the
following two conditions hold:
 \begin{itemize}
\item
At least $(s/16m)$ of $s$
strings $\tilde z_{r(1)},\ldots,\tilde z_{r(s)}$ \textup(corresponding to the  certificates
given by Merlin\textup)
are $(1/2+1/8m)$-approximations to $\hat u$.
\item
For every codeword $v$ of $\LDC_{n,\delta}$, if at least $s/16m$ of  $s$ strings
$\tilde z_{r(1)},\ldots,\tilde z_{r(s)}$  are
$(1/2+1/8m)$-appro\-xima\-tions to $v$, then $v$ is a candidate.
\end{itemize}
\end{lemma}
\textbf{Proof:} see Claims~17 and~18 in~\cite{BLvM}.

In our Arthur--Merlin protocol we use the parameters
$s$ and $\gamma$ from Lemma~\ref{lemma1}.

\textbf{Second part of the Arthur--Merlin protocol}.
Arthur does not need anymore to communicate with Merlin. Now he
composes the list of all codewords $v$ that are
$(1/2+1/8m)$-approximated by 
at least $s/16m$ of strings $\tilde z_{1},\ldots,\tilde z_{s}$.
From Lemma~\ref{lemma1} it follows that
with probability more than $2/3$ all strings in this list are candidates,
and the string $\hat u$ is included
in the list. The program $p'$
defined above can distinguish
$\hat u$ from other strings from the list.

Thus, Arthur can find $\hat u$ in polynomial time if he is given
$b, p$ and the following additional information: the index $i$,
the bit $r_i$, the mean value $\bar g$, 
and the distinguishing program $p'$. In fact, it is enough
to know not the exact value of $\bar g$ but only
an approximation to this number; this approximation must be
precise enough so that Arthur can find the integer part of $s\bar g$.
Thus, the required additional information contains only $O(\log n)$ bits.

Now we check that the described protocol of
generating $a$ satisfies the definition of $\CAM$-complexity.
The $\CAM$-program for $a$
consists of  (i) the truth tables of functions $f_1,\ldots, f_{i-1}$ constructed from
$\hat u(y|_{S_1})$, \ldots, $\hat u(y|_{S_{i-1}})$ for $u=a$
(this is the longest part of the program; we denoted it by $p$), 
(ii) the bit $r_i$ chosen
so that (\ref{approx-ineq}) is true, (iii) a rational $\gamma$ and an approximation to 
a rational $\bar g$, and (iv) Sipser's code $p'$
that distinguishes $\hat u$ between all ``candidates''.
The Arthur--Merlin protocol works as follows. Arthur chooses
at random strings $r(1),\ldots,r(s)$. Merlin provides $s\cdot (\bar g-\gamma)$
certificates corresponding to these $r(j)$. Arthur 
computes  $\tilde z_{1}, \ldots, \tilde z_{s}$ corresponding to the obtained
certificates and finds the list of all $\LDC_{n,\delta}$-codewords $v$ that are $(1/2+1/8m)$-approximated
by at least $s/16m$ of these $\tilde z_j$.
Then Arthur selects $\hat u$ from this list of strings using distinguishing program
$p'$, and computes $a=\LDC^{-1}_{n,\delta}(\hat u)$.

If Merlin is fair, this plan works OK with probability more than $2/3$
(Lemma~\ref{lemma1}).
If Merlin wants to cheat, he has two options: provide a list of certificates
such that the required string $\hat u$ is not approximated by $s/16m$
of  $\tilde z_{1}, \ldots, \tilde z_{s}$, 
or such that at least $s/16m$ of $\tilde z_j$ approximate some \emph{non-candidate} 
codeword $v$ (in these cases Arthur fails to select $\hat u$ using $p'$). 
However from Lemma~\ref{lemma1} it 
follows that for random $r(1), \ldots, r(s)$ both these ways of cheating are 
impossible with probability more than $2/3$.~\qed

\begin{center}
        \textbf{Acknowledgments}
\end{center}
        \nopagebreak
This article is based on several discussions and  reports presented at
the Kolmogorov seminar (Moscow). Preliminary versions appeared
as \cite{shen-dagstuhl} and \cite{musatov-diplom}. The authors
are grateful to all participants of the seminar for many useful
comments. The authors thank also anonymous referees for very detailed
comments and helpful suggestions leading to a significant 
revision of this paper.

\end{document}